\documentclass{article}

\usepackage{amsmath,amsfonts,amssymb,amsthm,cleveref,color}
\usepackage{enumerate}
\usepackage{ifthen}

\usepackage[]{color}

\newif\ifskip
\skiptrue

\newtheorem{theorem}{Theorem}
\newtheorem{definition}[theorem]{Definition}

\newtheorem{observation}[theorem]{Observation}

\newtheorem{remark}[theorem]{Remark}

\newtheorem{examples}[theorem]{Examples}

\newtheorem{proposition}[theorem]{Proposition}

\begin{document}
\title{Extensions of Courcelle's Theorem without Logic}
\author{Yuval Filmus and Johann A. Makowsky \\ \small Faculty of Computer Science \\ \small Technion-Israel Institute of Technology, Haifa, Israel}

\maketitle
\begin{abstract}
Courcelle's Theorem states that
on graphs $G$ of tree-width at most $k$ with a given tree-decomposition of size $t(G)$, 
graph properties $\mathcal{P}$ definable in Monadic Second Order Logic
can be checked in linear time in the size of $t(G)$.
Inspired by L. Lov\'asz' work using connection matrices instead of logic,
we give a generalized version of Courcelle's theorem which replaces the definability hypothesis
by a purely combinatorial hypothesis using a generalization of connection matrices.
This paper clarifies the role of logic in such theorems and displays their purely combinatorial assumption.

\end{abstract}

\newif\ifskip
\skiptrue
\newcommand{\cA}{\mathcal{A}}
\newcommand{\cF}{\mathcal{F}}
\newcommand{\cH}{\mathcal{H}}
\newcommand{\cL}{\mathcal{L}}
\newcommand{\cC}{\mathcal{C}}
\newcommand{\cD}{\mathcal{D}}
\newcommand{\cE}{\mathcal{E}}
\newcommand{\cG}{\mathcal{G}}
\newcommand{\cP}{\mathcal{P}}
\newcommand{\cT}{\mathcal{T}}
\newcommand{\cS}{\mathcal{S}}
\newcommand{\fA}{{\mathfrak A}}
\newcommand{\fB}{{\mathfrak B}}
\newcommand{\fC}{{\mathfrak C}}
\newcommand{\B}{{\mathbb B}}
\newcommand{\N}{{\mathbb N}}
\newcommand{\Z}{{\mathbb Z}}
\newcommand{\MSOL}{\mathbf{MSOL}}
\newcommand{\SOL}{\mathbf{SOL}}
\newcommand{\CMSOL}{\mathbf{CMSOL}}
\newcommand{\FOL}{\mathbf{FOL}}
\newcommand{\FPT}{\mathbf{FPT}}
\newcommand{\NP}{\mathbf{NP}}
\newcommand{\Str}{\mathrm{Str}}
\newcommand{\IND}{\mathrm{IND}}
\newcommand{\MW}{\mathrm{MW}}
\newcommand{\CW}{\mathrm{CW}}
\newcommand{\TW}{\mathrm{TW}}
\newcommand{\Subst}{\mathrm{Subst}}
\newcommand{\TWW}{\mathrm{TWW}}
\newcommand{\bF}{\mathbf{F}}
\newcommand{\bA}{\mathbf{A}}

\section{Introduction and outline}
\label{se:intro}

Monadic Second Order Logic ($\MSOL$) is an extension of First Order Logic which can be used to describe
formal languages and graph properties. We assume the reader is somehow familiar the basics of Finite Model theory, 
see \cite{ebbinghaus1995finite,libkin2004elements},
the notion of tree-width and other notions of graph-width, \cite{hlinveny2008width},
and Courcelle's Theorem, \cite{downey2012parameterized,downey2013fundamentals}.

\subsection{Three classical theorems}

\subsubsection*{Courcelle's Theorem}
Courcelle's Theorem  from 1990 is described in its Wikipedia page as follows, \cite{wi-courcelle}:
\begin{quote}
In the study of graph algorithms, Courcelle's theorem is the statement that every graph property definable in 
the monadic second-order logic of graphs can be decided in linear time on graphs of bounded tree-width. 
The result was first proved by Bruno Courcelle in 1990, \cite{courcelle1990monadic} 
and independently rediscovered by Borie, Parker and Tovey (1992), \cite{borie1992automatic}. 
It is considered the archetype of algorithmic meta-theorems.
\end{quote}
In \cite{courcelle2000linear,courcelle2001fixed} an analogue of Courcelle's Theorem was formulated and proved for graph classes of bounded clique-width, and
in \cite{courcelle2002fusion}, see also \cite[Theorem 4.21]{makowsky2004algorithmic}, a very general version was formulated and proved for $\CMSOL$, 
the extension of $\MSOL$ with added modular counting quantifiers. It also holds for certain recursively defined graph classes 
which generalize tree-width and clique-width.

Courcelle's Theorem is not the first meta-theorem involving MSOL.
It has a precursor in formal language theory.

\subsubsection*{The B\"uchi, Elgot, Trakhtenbrot Theorem}

The B\"uchi, Elgot, Trakhtenbrot Theorem from 1960-62, aka the BET-Theorem, is  summarized in its Wikipedia page, \cite{wi-BET}:
\begin{quote}
In formal language theory, the B\"uchi–Elgot–Trakhtenbrot Theorem states that a language is regular if and only if it 
can be defined by a formula in monadic second-order logic (MSOL). 
The theorem is due to Julius Richard B\"uchi, Calvin Elgot, and, independently, to  Boris Trakhtenbrot,
\cite{buchi1960weak,elgot1961decision,trakhtenbrot1962finite,trahtenbrot1966finite}.
\end{quote}

\subsubsection*{The Specker-Blatter Theorem}
It also has a precursor in finite combinatorics from 1981, see \cite{blatter1984recurrence}. 
\begin{quote}
The Specker-Blatter Theorem states that the number $s(\cP, n)$  of labeled graphs over
$n$ vertices satisfying an MSOL-deﬁnable property $\cP$ satisﬁes 
for every modulus $m$ a linear modular recurrence relation, which depends on $\cP$ and $m$. 
\end{quote}
It turns out that each of these theorems has a MSOL-free version, where definability in MSOL is replaced by an algebraic property
of a certain Hankel matrix.

\subsection{Hankel matrices}
Let $\cS$ be a class of finite $\tau$-structures and let $\Box$ be a binary relation on $\cS$.
Let $\cP \subseteq \cS$ be closed under $\tau$-isomorphisms. 

We define an infinite matrix $\cH(\cS,\Box, \cP)$, the Hankel matrix associated with $\cS,\Box$ and $\cP$ as follows:
\begin{enumerate}[(i)]
\item
Let $S_i$ be an enumeration of $\cS$ (up to $\tau$-isomorphism).
The rows and columns of $\cH(\cS,\Box, \cP)$ are labeled  by $S_i$.
\item
The entries of
$\cH(\cS,\Box, \cP)$ are $0$ or $1$ with $\cH(\cS,\Box, \cP)_{i,j} = 1$ iff $S_i \Box S_j \in \cP$.
\item
$\cH(\cS,\Box, \cP)$ can be viewed as a matrix over the finite field $GF(2)$ (equivalently the two element boolean algebra).
\item
We denote by $\rho(\cS,\Box, \cP)$  the rank of $\cH(\cS,\Box, \cP)$ in $GF(2)$.
\item
We denote by $\epsilon(\cS,\Box, \cP)$  the number of distinct rows of $\cH(\cS,\Box, \cP)$. This corresponds to the number of equivalence classes of 
of $\sim_{\Box}^{\cP}$ as defined in Section \ref{se:myproof} and used in Proposition \ref{pr:hankel}.
We call $\epsilon(\cS,\Box, \cP)$ the {\em equivalence number of $\cH(\cS,\Box, \cP)$}.
\end{enumerate}

$\rho(\cS,\Box, \cP)$ and $\epsilon(\cS,\Box, \cP)$ are closely related.
In the Courcelle-Lov\'asz Theorem the rank is used, in the Specker-Blatter Theorem the equivalence number is used.
\begin{proposition}
\label{obs:1}
Let $\cH(\cS,\Box, \cP)$ be given.
\begin{enumerate}[(i)]
\item
$\rho(\cS,\Box, \cP)$ is finite iff $\epsilon(\cS,\Box, \cP)$ is finite.
\item
$\rho(\cS,\Box, \cP) \leq \epsilon(\cS,\Box, \cP) \leq 2^{\rho(\cS,\Box, \cP)}$.
\end{enumerate}
\end{proposition}

Let us look at four binary operations on graphs $G = (V(G), E(G))$ with $E(G) \subseteq V(G)^2$:
Let $G, H$ be such graphs
\begin{enumerate}[(i)]
\item
The disjoint union of $ G= H_1 \sqcup H_2$ of  $H_1$ and $H_2$;
\item
The join union of $G = H_1 \bowtie H_2$ of  $H_1$ and $H_2$, which is the disjoint union of $V(H_1)$ and $V(H_2)$
with 
$$
E(G) = E(H_1) \cup E(H_2) \cup \{ e=(i,j) \in V(G)^2 : i \in V(H_1), j \in V(H_2) \}.
$$
\item
The tensor product of $G = H_1 \times_t H_2$ of  $H_1$ and $H_2$
with $V(G) = V(H_1) \times V(H_2)$ and 
$$
E(G) = \{ ((u_1, u_2), (v_1, v_2)) \in V(G) : (u_1, v_1) \in V(H_1) , (u_2, v_2) \in V(H_2) \}.
$$
\item
The Cartesian product of $G = H_1 \times_c H_2$ of  $H_1$ and $H_2$;
\begin{gather}
E(G) =  \notag \\
\{ (u_1, u_2, v_1, v_2) \in V(G)^4: 
u_1 = v_1  \wedge (u_2, v_2) \in V(H_2) 
\text{  or  }
u_2 = v_2  \wedge (u_1, v_1) \in V(H_1)\}.
\notag
\end{gather}
\end{enumerate}

We compute $\cH( \Box, \cD)$ for $\Box$ any of the above binary operations.
\begin{examples}
\label{ex:lowrank}
\begin{enumerate}[(i)]
\item
A graph $ G= H_1 \sqcup H_2$ is connected iff $H_1$ or $H_2$ is the empty graph.
Hence $\rho( \sqcup, \cD) =2$.
\item
A graph $ G= H_1 \bowtie H_2$ is connected iff $H_1$ is connected and $H_2$ is the empty graph (or vice versa) or
both are not empty.
Hence $\rho( \sqcup, \cD) =3$.
\item
A graph $ G= H_1 \times_t H_2$ is connected iff  
one is empty or both are connected and one is non-bipartite.
Hence $\rho( \sqcup, \cD) =4$. 
\item
A graph $ G= H_1 \times_c H_2$ is connected iff  
one is empty or both are connected.
Hence $\rho( \sqcup, \cD) =3$. 
\end{enumerate}
\end{examples}
Let $\cD$ be any class of connected graphs containing graphs of every order $n \in \N^+$. 

\begin{observation}
\label{obs:2}
We note that there are continuum many such classes,
but there are  only countably many classes of $\tau$-structures which are $\CMSOL$-definable.
\end{observation}

\begin{theorem}
\label{th:lowrank}
Let $\Box \in \{\sqcup, \bowtie, \times_t, \times_c\}$.
There are 
continuum many graph properties $\cP$ of graphs with $\rho(\cG \Box, \cP) \leq 3$.
\end{theorem}
\begin{proof}
We only prove it for $\Box = \sqcup$, the other cases being similar.

Let $\cD, \cD'$ be two distinct classes of connected graphs containing graphs of every order $n \in \N^+$. 

\begin{enumerate}[(i)]
\item
By observation \ref{obs:2} there are continuum many such classes.
\item
$\cH(\cG, \Box,\cD)$  and  $\cH(\cG, \Box, \cD')$
both have the same number of distinct rows.
One labeled with the empty graph and one labeled with any non-empty graph. 
hence
$$
rho(\cG, \Box,\cD) = \rho(\cG, \Box, \cD') \leq 2.
$$
\end{enumerate}
\end{proof}
More examples with other binary operations on graphs can be found in \cite[Section 13]{fischer2011application}.


\subsection{Replacing $\CMSOL$ by finite Hankel rank}
\subsubsection*{The B\"uchi-Elgot-Trakhtenbrot Theorem}

For the B\"uchi-Elgot-Trakhtenbrot Theorem we let $\cS = \Sigma^*$ be the words over a finite alphabet $\Sigma$, and $\Box = \circ$ be the concatenation of
words. The Hankel matrix $\cH(\Sigma^*,\circ, \cL)$ is sometimes also called the {\em concatenation matrix}.

Now the BET-Theorem can be stated as follows:
\begin{theorem}
Let $\cL \subseteq \Sigma^*$ be a language. Then $\cL$ is regular iff $\epsilon(\Sigma^*, \circ, \cL)$ is finite.
\end{theorem}
This was first formulated using Hankel matrices in \cite{carlyle1971realizations} where a similar theorem was proven for stochastic (or weighted) automata.
However, the version above follows readily from the classical Myhill-Nerode Theorem, see \cite{hopcroft1969formal}.

\subsubsection*{The Specker-Blatter Theorem}
If a sequence of integers satisfies for every $m$ a linear recurrence relation modulo $m$ which depends on $m$,
we call it MC-finite (modularly C-finite), in analogy to C-finite sequences over $\Z$.

For the Specker-Blatter Theorem we let $\tau$ be a finite vocabulary consisting of relation symbols of arity at most $2$.
We let $\cS = Str(\tau)$ be  the class of all $\tau$-structures. Let $(\fA,a')$ be a $\tau$-structure with a distinguished element $a'$.
We define $\fC= Subst((\fA,a'), \fB)$ the structure where the element $a'$ is substituted by $\fB$.
More precisely, the universe $C$ of $\fC$ is the disjoint union of the universes $A, B$ of $\fA, \fB$ with $a$ removed.
The binary relations $R \in \tau$ restricted to $A-\{a\}$ and $B$ remain the same. For $(a,b), a \in A, b\in B$ we define
$(a,b) \in R$ iff $(a,a') \in R$. In other words, $\fB$ is a module (in the sense of graph theory) of $\fC$. 
Let $\cP$ a class of $\tau$-structures.
The Hankel matrix $\cH(Str(\tau), \Subst, \cP)$ is sometimes also called
the {\em substitution matrix}.

\begin{proposition}
Assume  $\cP$ is MSOL-definable. Then $\epsilon(Str(\tau),\Subst, \cP)$ is finite.
\end{proposition}

It was already realized in \cite{blatter1984recurrence}  that definability  in MSOL of $\cP$ was sufficient, but not necessary.
\begin{theorem}
Assume $\rho(Str(\tau),\Subst, \cP) =r$ is finite. 
Then $s(\cP, n)$ is MC-finite and
the linear recurrence relation modulo $m$ depends on $r$ and $m$.
\end{theorem}

The reader interested in the Specker-Blatter Theorem and its applications can consult
\cite{fischer2011application,filmus2023mc,fischer2024extensions,filmus2026effective}. However, for the sequel of this paper this is not needed.

\subsubsection*{The Courcelle-Lov\'asz Theorem}

In 2005, L. Lov\'asz, \cite{freedman2005graph,lovasz2006rank,lovasz2007connection}, introduced another interesting class of Hankel matrices, 
which he called {\em connection matrices}.
They play an important role in the theory of graph limits, \cite{lovasz2012large} and Finite Model Theory, in particular in proving
that graph parameters are not definable in $\CMSOL$ and some of it sublogics.
\cite{godlin2008evaluations,kotek2014connection}.

Here the structures $Str(\tau) = \cG_k$ consist of $k$-graphs 
$$
G_k = (V(G), E(G), a_1, \ldots a_k),
$$ 
graphs with $k$ distinctive constants, called {\em ports}.
Given two $k$ graphs $G_k, H_k$ we define a binary operation $G_k \sqcup_k H_k$ by taking the disjoint union of the underlying graphs and then identifying
the corresponding ports.
Let $\cP$ be a property of $k$-graphs. The Hankel matrix $\cH(\cG_k, \sqcup_k, \cP)$ is called the {\em connection matrix for $\cP$}.

It is shown in \cite{godlin2008evaluations} that:
\begin{theorem}[Finite Rank Theorem]
For $\cP$ definable in $\CMSOL$  the rank $\rho(\cG_k, \sqcup_k, \cP)$ of the connection matrix $\cH(\cG_k, \sqcup_k, \cP)$ is
always finite.
\end{theorem}

In \cite[Theorem 6.48]{lovasz2012large} the following is stated and a proof is sketched.
\begin{theorem}[Courcelle-Lov\'asz Theorem]
\label{th:C-L}
Let $\cC$ be a class of $k$-graphs of tree-width at most $k$ and $\cP$ with $\rho(\cG_k, \sqcup_k, \cP) =r$ for some $r \in \N$.
Then checking whether a graph $G \in \cC$ with tree-decomposition $t(G)$  satisfies property $\cP$ is 
fixed parameter tractable (in $\FPT$) in the parameter $k$ in the size of $t(G)$.
\end{theorem}
We shall discuss his proof briefly in Section \ref{se:myproof}.

The purpose of this paper is to formulate and prove an analogue of Theorem \ref{th:C-L} for a large class of recursively defined graph classes including
graph classes of bounded clique-width, modular width, and others.

\subsection{Outline of the paper}

In Section \ref{se:myproof} we prove Theorem \ref{th:C-L} and its generalization.
In Section \ref{se:applications} we show that the class of graphs of tree-width at most $k$, of clique-width at most $k$ and of modular width
at most $k$ can be defined as inductive classes.
In Section \ref{se:beyond} we discuss two cases where we do not know how to apply our proof methods: The case of graph classes of bounded twin-width,
and the general case of hypergraphs and matroids.
In Section \ref{se:conclu} we discuss what we have achieved and give directions for further research.

\section{Main Theorem}
\label{se:myproof}

In this section we formulate and prove a generalization of Theorem \ref{th:C-L}.
However, our proof is different from the proof in \cite[Chapter 6]{lovasz2012large}.
There the proof is algebraic and specially tailored for connection matrices.
Our approach follows more the ideas presented in \cite{courcelle2002fusion,makowsky2004algorithmic}.

Operations on $\tau$-structures are {\em invariant under $\tau$-isomorphisms}.
The notion of an {\em inductively defined class of $\tau$-structures} is taken form \cite{makowsky2004algorithmic}.
We first show the main idea for inductively defined classes with exactly one binary operation on $\tau$-structures.

\subsection{One binary operation on $\tau$-structures}
\begin{definition}[Inductively defined class of $\tau$-structures]
\ \\
Let $\cA$ be a finite set of $\tau$-structures, $\Box$ be a binary operation on $\tau$-structures, and
$\cF = \{F_1, \ldots, F_s\}$ be a finite set of unary operation on $\tau$-structures.

A class $\IND(\cA, \Box, \cF)$ of $\tau$-structures is  defined {\em  inductively} as follows:
\begin{enumerate}[(i)]
\item
$\cA \subseteq \IND(\cA, \Box, \cF)$.
\item
If $\fA, \fB \in \IND(\cA, \Box, \cF)$  so is $\fA \Box \fB$.
\item
If $F \in \cF$ and  $\fA \in \IND(\cA, \Box, \cF)$  so is $F(\fA)$. 
\end{enumerate}
If both the binary operation on $\tau$-structures $\Box$, and the unary operations in $\cF$ are computable 
in linear time (in the size of the $\tau$-structure), we say that the inductively defined class of $\tau$-structures
{\em effectively inductively} defined.
\end{definition}

Let $\cS$ a class of $\tau$-structures, and let $\cP \subseteq \cS$ be a property of $\tau$-structures.
We look at the Hankel matrix $\cH = \cH(\cS, \Box, \cP)$.

We define an equivalence relation $\sim_{\Box}^{\cP}$ on $\cS$ by
\begin{gather}
\fA \sim_{\Box}^{\cP} \fB \text{ iff }  \notag \\
\forall \fC \text{ we have } 
\fA \Box \fC \in \cP \text{ iff }  \fB \Box \fC \in \cP 
\notag
\end{gather}
For a $\tau$-structure $\fA$ we denote by $[\fA]_{\Box}^{\cP}$ its equivalence class.
We denote by $0_{\Box}$ a graph such that $0_{\Box} \Box \fA = \fA$ for all $\tau$-structures $\fA \in \cS$.
$0_{\Box}$ is called  a {\em neutral element for $\Box$}.

\begin{definition}[Smooth operations]
\ 
\begin{enumerate}[(i)]
\item
$\Box$ is smooth for $\cP$ if $\fA_1 \sim_{\Box}^{\cP} \fA_2$ and $\fA'_1 \sim_{\Box}^{\cP} \fA'_2$ then also $\fA_1 \Box \fA_2 \sim_{\Box}^{\cP} \fA'_1 \Box \fA'_2$.
\item
If $F$ is a unary operation on $\tau$-structures, $F$ is smooth for $\cP$ if  $\fA_1 \sim_{\Box}^{\cP} \fA_2$ then $F(\fA_1) \sim_{\Box}^{\cP} G(\fA_2)$.
\item
$\IND(\cA, \Box, \cF)$ is smooth for $\cP$, if $\Box$ and every $F \in \cF$ are smooth for $\cP$
and $\IND(\cA, \Box, \cF)$ has a neutral element for $\Box$.
\end{enumerate}
\end{definition}

\begin{remark}
Let $\cG$ be the class of graphs and $\Box = \times$ be the Cartesian product of graphs, $\cF$ a finite set of unary graph operations.
ad $\cP$ be a graph property.
$\IND(\cG, \times, \cF)$ and its Hankel matrix are  well defined. 
The graph consisting of one isolated vertex is a neutral element. 
However, 
$\IND(\cG, \times, \cF)$ is not effective, because $\times$ is not computable in linear time.
\end{remark}

\begin{proposition}
\label{pr:hankel}
We denote by  $\cH_{\fA}$ the row of $\cH$ labeled by $\fA$.
\begin{enumerate}[(i)]
\item 
$\fA \sim_{\Box}^{\cP} \fB$ iff $\cH_{\fA} =\cH_{\fB}$.
\item 
Assume $0_{\Box}$ exists.  Then $\fA \in \cP$ iff $\cH_{0_{\Box}, \fA} = 1$.
\item 
By Proposition \ref{obs:1},
$\rho = \rho(\cS, \Box, \cP) \leq \epsilon(\cS, \Box, \cP)$, i.e., it is at most the number of $\sim_{\Box}^{\cP}$-equivalence classes.
\item 
If the equivalence relation on $\tau$-structures $\equiv_{\Box}^{\cP}$ is a refinement of $\sim_{\Box}^{\cP}$
and $\equiv_{\Box}^{\cP}$ is smooth for $\Box$ and $\cP$, so is $\sim_{\Box}^{\cP}$.
\end{enumerate}
\end{proposition}

\begin{remark}
If $\cP$ is definable in $\CMSOL$ by a formula $\phi$
we look at the quantifier rank $q$ of $\phi$.
The equivalence relation $\sim_{\Box}^{\cP}$ then corresponds to satisfying the same $\tau$-sentences of quantifier rank $q$.
\end{remark}

Given an inductively defined class of $\tau$-structures $\cS = \IND(\cA, \Box, \cF)$ and $\cP$ a property of structures in $\cS$.
Assume $\rho = \cH(\cS, \Box, \cP)=m$, and let $\fA_1, \ldots, \fA_m$ be pairwise nonequivalent. In the Hankel matrix $\cH(\cS, \Box, \cP)$
the rows labeled $\fA_1, \ldots, \fA_m$ form a basis for the Hankel matrix.

Let $\alpha: \cS \rightarrow \{\fA_1, \ldots, \fA_m\}$ be a function such that $\alpha(\fA) = \fA_i$ for $\fA \sim_{\sim_{\Box}}^{\cP} \fA_i$.

\begin{definition}[Look-up table $(\hat{\cH}, \hat{\cF})$] 
Given an inductively defined class of graphs $\cS =\IND(\cA, \Box, \cF)$ of $\tau$-structures and $\cP \subset \cS$  a property of $\tau$-structure.
Let $\rho = \cH(\cS, \Box, \cP)=m$. We construct a look-up table as follows:
\begin{enumerate}[(i)]
\item
We compute $\alpha(\fA)$ for every $\fA \in \cA$. This has at most $m$  distinct elements.
\item
For $\Box$ we build an $m \times m$-matrix $\hat{\cH}$ whose rows and columns  and rows are labeled $\\fA_1, \ldots, \\fA_m$ and
$\hat{\cH}_{\alpha_i, \alpha_j} = \alpha(\fA_i \Box \fA_j)$.
This has at most $m \times m$ distinct elements.
\item
Let $\cF=\{F_1, \ldots, F_n\}$. We build an $m \times n$-matrix $\hat{\cH}$ where the rows are labeled by  $\{\fA_1, \ldots, \fA_m\}$ 
and the columns are labeled by $\{F_1, \ldots, F_n\}$.
The entries of $\hat{\cF}$ are  
$\hat{\cF}_{\alpha_i, F_j} = \alpha(F_j(\fA_i))$.
This has at most $m \times n$ distinct elements.
\end{enumerate}
\end{definition}

\begin{proposition}
\ \\
Let $\rho = \cH(\cS, \Box, \cP)=m$.
Assume the inductive class of structures $\cS = \IND(\cA, \Box, \cF)$ is smooth for $\Box$ and $\cP$.
Then the look-up table  $(\hat{\cH}, \hat{\cF})$ is well defined and depends only on the equivalence classes $[\fA]_{\Box}^{\cP}$ for $\fA \in \cS$.
\end{proposition}

Let $\cS = \IND(\cA, \Box, \cF)$ and $\fA \in \cS$. Let $T(\fA)$ be a {\em parse-tree of the inductive definition of $\fA$} 
where the leaves are in $\cA$
and the ancestors of $\fB_1, \fB_2$ are $\fB_1 \Box \fB_2$ or $F_1(\fB_1)$ and $F_2(\fB_2)$ for some $F_1, F_2 \in \cF$ and the root is $\fA$.

\begin{definition}[The Main Algorithm]
\label{algorithm}
\ \\
\begin{description}
\item[Preprocessing:]
Compute the look-up table. This is a finite object.
\item[Compute $\alpha(\fA)$:]
Use the parse tree of $\fA$. For the leaves use the look-up table.
\\
For an ancestor $\fB = \fB_1 \Box \fB_2$ use the fact that $\Box$ is computable in linear time, then use the look-up table to compute
$\alpha(\fB)$. Similarly for the operations from $\cF$.
\item[Check if $\fA \in \cP$:]
Use the look-up table to check whether $\cH_{0_{\Box}, \alpha(\fA)} = 1$.
\end{description}
\end{definition}

\begin{theorem}[Main Theorem]
\label{th:main}
Let $\cP$ be a property of $\tau$-structures.
Assume the inductive class of structures $\cS = \IND(\cA, \Box, \cF)$ is effective and smooth for $\Box$ and $\cP$,
and that $\cH = \cH(\cS, \Box, \cP)$ has finite rank $\rho(\cH)=m$,
\begin{enumerate}[(i)]
\item
Then there exists a finite look-up table.
\item
If a $\tau$-structure $\fA \in \cS$  is given with its parse-tree $T(\fA)$,
checking whether  $\fA \in \cP$
is in $\FPT$ in the parameter $m$.
\end{enumerate}
\end{theorem}

To prove Theorem \ref{th:C-L} from Theorem \ref{th:main} we have to show that the class of graphs ($k$-graphs) of tree-width at most $k$
has an equivalent definition as an effective and smooth inductive  class.
Instead of $k$-graphs we use colored graphs with at most $k$ colors with $\Box$ the disjoint unions and fusion of colors instead
of gluing. This was done in detail in \cite{courcelle2002fusion} and in \cite[Section 4.21]{makowsky2004algorithmic}.
We will repeat the proof in Section \ref{se:applications}.

\subsection{Several binary operations on $\tau$-structures}

Let $\B = \{ \Box_1, \ldots, \Box_s\}$ be a finite set of binary operations on $\tau$-structures.
The class $\cS  = \IND(\cA, \B, \cF)$ of $\tau$-structures is {\em inductively defined analogously}.
We want to prove the analogue of Theorem \ref{th:main} for this case.

For $i \in [s]$ We look at the Hankel matrices $\cH_{\B_i}$ for each $\Box_i \in \B$ separately.
For a property of $\tau$-structures $\cP$ let $\sim_{\Box_i}^{\cP}$ be the equivalence relation on $\tau$-structures
defined as
\begin{quote}
$\fA_1 \sim_{\Box_i}^{\cP} \fA_2$ iff for all $\tau$-structures $\fC \in \cS$
\\
$\fA_1 \Box_i \fC \in \cP$ iff $\fA_2 \Box_i \fC \in \cP$. 
\end{quote}
Even if all the Hankel matrices $\cH_{\B_i}$ for $\Box_i \in \B$ have finite rank they their respective equivalence
relations $\sim_{\Box_i}^{\cP}$ may be incompatible.
We therefore  introduce a new equivalence relation, which is the coarsest joint refinement of the equivalence relations
$\sim_{\Box_i}^{\cP}$.

\begin{definition}
The equivalence relation $\sim_{\B}^{\cP}$ on $\cS$ is the intersection of the equivalence relations
$\sim_{\Box_i}^{\cP}$ for $i \leq s$.
\end{definition}

\begin{observation}
\label{obs:3}
If all the equivalence relations $\sim_{\Box_i}^{\cP}$ have only finitely many equivalence classes,
so does $\sim_{\Box_i}^{\cP}$. 
\end{observation}

\begin{definition}
{\em $\cS  = \IND(\cA, \B, \cF)$ is smooth for $\B$ for and $\cP$}
if $\sim_{\B}^{\cP}$ is smooth for $\B$ for and $\cP$.
$\cS  = \IND(\cA, \B, \cF)$ is effective if all the operations $\Box_i \in \B$ and $F_i \in \cF$ or computable in linear time.
Furthermore we require that for at least one $i \leq s$ there is a neutral element $0_{\Box_i}$. 
\end{definition}

\begin{observation}
Assume that $\sim_{\B}^{\cP}$ is smooth for $\B$ for and $\cP$. Then $\sim_{\Box_i}^{\cP}$ is smooth for every $i \in [s]$,
by Proposition \ref {obs:1}.
\end{observation}

Assume $\sim_{\Box}^{\cP}$ has exactly $t \in \N^+$ equivalence classes and
let $\{\fA_1, \ldots, \fA_t\}$ be mutually nonequivalent $\tau$-structures.
Let $\alpha$ be the function which maps a $\tau$-structure $\fA \in \cS$ to its equivalence class $[\fA] = [\fA_j]$ for some $j \in [t]$.
$\alpha$ can be computed from Hankel matrices $\cH_{\B_i}$, $i \in [s]$ by inspecting their rows.

We build $s$ look-up tables $\hat{\cH}^k, k \in [s]$. Their rows and columns are labeled by $\{\fA_1, \ldots, \fA_t\}$.
The entry $\hat{\cH}i_{\fA_i, \fA_j}$ is given by $\alpha(\fA_i \Box_k \fA_j)$.
The modified Hankel matrices now have all $t$ rows and columns.

With these definitions the look-up tables are defined and the algorithm from Definition \ref{algorithm} can be adapted.

\begin{theorem}[Main Theorem, general case]
\label{th:main-a}
Let $\cP$ be a property of $\tau$-structures.
Assume the inductive class of structures $\cS = \IND(\cA, \B, \cF)$ is effective and smooth for $\B$ and $\cP$.
and that $\cH = \cH(\cS, \B, \cP)$ has finite rank $\rho(\cH)=m$,
\begin{enumerate}[(i)]
\item
Then there exists a finite look-up table.
\item
If a $\tau$-structure $\fA \in \cS$  is given with its parse-tree $T(\fA)$,
checking whether  $\fA \in \cP$
is in $\FPT$ in the parameter $m$.
\end{enumerate}
\end{theorem}

\begin{remark}
If all the operations  of $\cS  = \IND(\cA, \B, \cF)$ are smooth for $\CMSOL$ and $\cP$ is definable in $\CMSOL$,
Theorem \ref{th:main-a} corresponds to the corresponding Theorem from \cite[Section 4.21]{makowsky2004algorithmic}.
\end{remark}

\section{Applications}
\label{se:applications}
We illustrate how to  use Theorem \ref{th:C-L} in the cases of tree-width, clique-width, and modular width.

\subsection{Tree-width}

We define here $\TW(k)$, the class of graphs of tree-width at most $k$ inductively.
This approach is taken from \cite{makowsky2004algorithmic}.
We do it on colored graphs $\bar{G}= (V(G), E(G), C_1, \ldots, C_k)$ augmented with $k$ unary predicates the interpretation of which are disjoint possibly empty sets.
Here, the colors correspond to ports and fusion of ports gives $k$-connections.
The operation $\mathrm{fuse}_i(G)$ contracts the set of elements $C_i(G)$  of $V(G)$ colored with color $i$ to a single point $c$.
If an element $c' \in C_i(G)$ and $(u, c') \in E(G)$  then $(u,c) \in E(\mathrm{fuse}_i(G))$, and the same for $(c', u)$.

\begin{enumerate}[(i)]
\item 
All graphs $\bar{G}$ with at most $k+1$ vertices are in $\TW(k)$. 
\item 
$TW_k$ is closed under disjoint union.
\item 
$TW_k$ is closed under renaming of colors.
\item 
$TW_k$ is closed under {\em fusion}, i.e. for every 
colored graph $G \in TW_k$,
and for every unary predicate symbol representing a vertex color $C_i$, 
also $\mathrm{fuse}_i(\bar{G}) \in TW_k$.
\end{enumerate}
This definition is not the standard definition given, 
say in \cite{diestel2000graduate},
but is equivalent to it, cf. \cite{courcelle2002fusion}.

\begin{proposition}
Let $\bar{G}$ be a colored graph with $k$ colors. Then $\bar{G} \in \TW(k)$ iff $\bar{G}$ has tree-width at most $k$.
\end{proposition}
\begin{proof}
In contrast to the standard definition, we have no ports but colors.
We proceed by induction. 
Colored graphs $\bar{G}$ with at most $k+1$ vertices have tree-width at most $k$.
Disjoint union preserves tree-width. Recoloring does not affect tree-width.
Finally, fusion does not increase the tree-width.
Hence all colored graphs in $ \TW(k)$ have tree-width at most $k$.

Conversely, assume $G$ has tree-width at most $k$ and let a tree-decomposition of $G$.
It may be useful to use smooth $k$ tree-decompositions. A $k$ tree-decomposition is smooth 
if every bag has size $k+1$ and the intersection of two neighboring bags has size $k$.
We turn it into a tree decomposition of
$\bar{G}$ by coloring it inductively.
\end{proof}

\begin{theorem}
\label{th:TW-abstract}
\begin{enumerate}[(i)] 
\item
Let $\bar{G} \in \TW(k)$ with a parse-tree $t(\bar{G})$ of size $m(\bar{G})$ and $P$ be a $\CMSOL$-definable property of colored graphs.
Then checking $\bar{G} \in P$ is in linear time  in the size of the parse-tree $m(\bar{G})$.
\item
Let $\bF$ consist of disjoint union, renaming colors $\rho_{i,j}$ and fusion.
Let $P$ be a graph property such that the infinite $\bF$-circuit matrix $CM_{\bF,P}$ is of finite rank, 
then checking $G \in P$ is in linear time  in the size of the parse-tree $m(G)$.
\end{enumerate}
\end{theorem}

\subsection{Clique-width}
We define here $\CW(k)$, the class of graphs of clique-width at most $k$ inductively.
We again define it for colored graphs with at most $k$ colors.
 
\begin{enumerate}[(i)]
\item 
All colored graphs $\bar{G}$ with one vertex colored with one of the colors are in $CW_k$.
\item 
$CW_k$ is closed under disjoint union.
\item 
$CW_k$ is closed under renaming of colors.
\item 
$CW_k$ is closed under $\eta_{i,j}(\bar{G})$ which adds all the edges between the two colors $i$ and $j$.
\end{enumerate}

\begin{theorem}
\label{th:CW-abstract}
\begin{enumerate}[(i)] 
\item
Let $\bar{G} \in \CW(k)$ with a parse-tree $t(\bar{G}$ of size $m(\bar{G})$ and $P$ be a $\CMSOL$-definable property of colored graphs.
Then checking $\bar{G} \in P$ is in linear time  in the size of the parse-tree $m(\bar{G})$.
\item
Let $\bF$ consist of disjoint union, renaming colors $\rho_{i,j}$ and 
under $\eta_{i,j}(\bar{G})$ which adds all the edges between the two colors $i$ and $j$.
Let $P$ be a graph property such that the infinite $\bF$-circuit matrix $CM_{\bF,P}$ is of finite rank, 
then checking $G \in P$ is in linear time  in the size of the parse-tree $m(G)$.
\end{enumerate}
\end{theorem}

\subsection{Modular-width}

We consider various substitution operations for $k$-graphs 
$$
G = (V(G), E(G), v_1, \ldots , v_k).
$$
\begin{definition}
Let $H$ be a graph on $k$ vertices and $G_1, \ldots, G_k$ be graphs $G_i =(V(H_i), E(H_i))$.
The graph $\bar{H} = H[G_1, \ldots, G_k]$ is defined as follows:

\begin{enumerate}[(i)]
\item
$V(\bar{H}) = \bigcup_i^k V(H_i)$.
\item
$E(\bar{H}) = \bigcup_i^k E(H_i) \cup \{(u, v): u \in V(H_i), v \in V(H_j), i \neq j, (v_i, v_j) \in E(G) \}$
\\
Hence $\bar{H} = H[G_1, \ldots, G_k]$ is obtained from $G$
by substituting every vertex $v_i \in V (G)$ with the graph $H_i$ and adding all edges between the
vertices of the graph $H_i$ and the vertices of a graph $H_j$ whenever $(v_i , v_j) \in E(G)$.
\end{enumerate}
We do not require that the vertex sets $V(H_i)$ are disjoint.
\end{definition}

In \cite{gajarsky2013parameterized} an inductive definition of the class of graphs $\MW(k)$ of modular-width at most $k$  is given as follows:
\begin{enumerate}[(i)] 
\item
Graphs consisting of one vertex  are in $\MW(k)$.
\item
If $G_1, G_2 \in \MW(k)$  then $G_1 \sqcup G_2 \in  \MW(k)$, i.e.,  $\MW(k)$
is closed under disjoint unions.
\item
If $G_1, G_2 \in \MW(k)$  then $G_1 \bowtie G_2 \in  \MW(k)$, i.e.,  $\MW(k)$
is closed under the join.
\item
Let $H$ be a graph with exactly $k$ vertices, and let $G_1, \ldots, G_k \in \MW(k)$.
Then $H[G_1, \ldots, G_k] \in \MW(k)$.
\end{enumerate}

\begin{proposition}
\begin{enumerate}[(i)] 
\item
There are only finitely many graphs on $k$ vertices.
Hence there are only finitely many $k$-ary operations
$H[G_1, \ldots, G_k]$.
\item
$H[G_1, \ldots, G_k]$ is $\MSOL$-smooth and also $\CMSOL$-smooth.
\end{enumerate}
\end{proposition}

\begin{theorem}
\label{th:MW-abstract}
\begin{enumerate}[(i)] 
\item
Let $G \in \MW(k)$ with a parse-tree $t(G)$ of size $m(G)$ and $P$ be a $\CMSOL$-definable graph property.
Then checking $G \in P$ is in linear time  in the size of the parse-tree $m(G)$.
\item
Let $\bF$ consist of disjoint union, join, and $H$-substitution for all graphs $H$ on $k$ vertices.
Let $P$ be a graph property such that the infinite $\bF$-circuit matrix $CM_{\bF,P}$ is of finite rank, 
then checking $G \in P$ is in linear time  in the size of the parse-tree $m(G)$.
\end{enumerate}
\end{theorem}

\section{Beyond inductive classes}
\label{se:beyond}

Our main theorem is proved using two ingredients: The fact that our notion of width has an inductive definition and that the operations used
or the inductive definition are smooth. In this section we briefly discuss two cases of possible applications of the main theorem,
where it is not clear whether these ingredients are available.

\subsection{Twin-width}
\label{se:twinwidth}
Recently a new notion of graph width, {\em twin-width} has been defined \cite{bonnet2021twin,bonnet2022twin} and found very useful.
It generalizes most of the currently used notions of width in graph theory.
It is, however, not clear how to give a recursive definition of graphs of twin-width at most $k$ which fits our present framework.
Currently no such definition is known and it seems unlikely that such a definition exists, \cite{Bonnet-private}.
In particular a logic based version of Courcelle's Theorem 
for twin-width
has only been shown 
for graph properties definable in first order logic $\FOL$, \cite{bonnet2021twin,gajarsky2022twin}.  
For graph properties definable in
monadic second order logic $\MSOL$ it has been shown for graph classes with {\em bounded contraction sequences}, a notion slightly weaker than
bounded twin-width, \cite{bonnet2021twin,bonnet2022twin}
For a detailed discussion the reader may consult \'E. Bonnet's habilitation thesis \cite{bonnet2024twin}.

\subsection{Hypergraphs and matroids}
There are various analogues of Courcelle's Theorem for certain classes of matroids.
In \cite[Section 3]{mayhew2026monadic}, the state of art as of 2026 is summarized.
Recent developments may be found in \cite{funk2022tree,funk2023tree}.

There are several ways of defining $\CMSOL$ for matroids depending whether they are defined as 
hypergraphs ($\CMSOL^{hyp}$), 
by independent set ($\CMSOL^{hyp}$), or any of the other ways considered in the Matroid literature, \cite{oxley2006matroid}.
Petr Hlin\v{e}n\'y, \cite{hlinveny2006branch} was the first to prove a Courcelle-type theorem for matroids:

Let $\mathbb{F}$ be a finite field.
\begin{theorem}[Hlin\v{e}n\'y 2005]
Let $\phi$ a sentence of $\CMSOL^{hyp}$.
There as an $\FPT$ algorithm for testing whether a $\mathbb{F}$-representable matroid of fixed branch-width satisfies $\phi$.
\end{theorem}

The proof of Hlin\v{e}n\'y's Theorem and many of its generalizations use that $\CMSOL$ properties of matroids can be recognized  by tree-automata,
and a version of the corresponding Myhill-Nerode Theorem. Therefore, it should be possible, at least for some cases of matroid classes,
to formulate and prove a version of our main theorem for matroids.

\section{Conclusions}
\label{se:conclu}

B. Courcelle's famous theorem 
 shows that many graph problems which are $\NP$-hard in general,
are in $\FPT$ when restricted to graph classes of bounded tree-width.
This theorem uses logic, more precisely definability in $\MSOL$, as its main hypothesis.
It is therefore a meta-theorem in the sense of \cite{makowsky2025matrix}.
L. Lov\'asz showed in \cite[Theorem 6.48]{lovasz2012large} that the hypothesis of definability in $\MSOL$ can be replaced by a purely combinatorial
hypothesis using connection matrices of finite rank.
A precursor of such a separation may be found in the pioneering paper from 1984 by C. Blatter and E. Specker on integer sequences \cite{blatter1984recurrence}.
There the hypothesis of definability in $\MSOL$ can be replaced by substitution matrices 
of finite rank. Substitution matrices are a special case of Hankel matrices for the $\MSOL$-smooth operation $G[H]$ of substitution of a graph $H$
into a pointed graph $G$. Equivalently, $H$ is a module in $G[H]$.
A modern treatment can be found in \cite{filmus2023mc}. 

A precursor of such a separation may be found in \cite{blatter1984recurrence} where definability in $\MSOL$ can be replaced by substitution matrices 
of finite rank. Substitution matrices are a special case of Hankel matrices.

An abstract version of Courcelle's Theorem which still involved logic can be found already in
\cite{makowsky2004algorithmic}.
Inspired by this we showed 
in this paper how to replace logic in Courcelle's Theorem 
by a purely combinatorial hypothesis.

There are several advantages in replacing definability in logic by a combinatorial condition.
\begin{itemize}
\item
Conceptual: It clarifies the role of logic and of combinatorics. Logic is only used as an easy way to show that the rank of the Hankel matrices for $\cP$ are finite.
\item
Algorithmic: The circuit rank of  $\cP$ is very often very much smaller than what one obtains from definability assumption.
\item
Range of applicability: There are continuum many properties with fixed Hankel rank, whereas there are only countable many $\MSOL$-definable properties $\cP$.
\end{itemize}

However, we do not know how to compute (or even give an upper bound for)  the Hankel rank or the circuit rank for $\cP$ in general.
It remains a challenging project to investigate what one has to know about $\cP$ in order to estimate or compute the Hankel or circuit rank.
Simple examples are given in Examples \ref{ex:lowrank} and Theorem \ref{th:lowrank}.

Courcelle's Theorem was also formulated for numeric graph parameters and graph polynomials. Our approach also extends to this settings.
In this case it is helpful to replace graphs ($\tau$-structures) by formal linear combinations of finitely many graphs or structures.
Instead of the field $GF(2)$ one uses now the underlying field or ring of the graph parameter and use as rank the rank over the field or ring.

There are many ways in which to pursue our line of research and results. Let us conclude with three major challenges:

\begin{description}
\item[Twin-width:] Can our results extended to graph classes bounded twin-width.
\item[Hypergraphs and matroids:] Formulate a logic-free version of Courcelle-like theorems for hypergraphs and matroids.
\item[Relational structures:] Find interesting versions of logic-free version of Courcelle-like theorems for other classes of relational $\tau$-structures.
\end{description}

\bibliography{St}
\bibliographystyle{alpha}
\end{document}